\documentclass[pre,10pt,twocolumn]{revtex4}%
\usepackage{amsfonts}
\usepackage{amsmath}
\usepackage{amssymb}
\usepackage{graphicx}%
\newtheorem{theorem}{Theorem}

\newtheorem{corollary}[theorem]{Corollary}

\newtheorem{definition}[theorem]{Definition}

\newtheorem{lemma}[theorem]{Lemma}
\newtheorem{fact}[theorem]{Fact}

\newtheorem{proposition}[theorem]{Proposition}

\newenvironment{proof}[1][Proof]{\noindent\textbf{#1.} }{\ \rule{0.5em}{0.5em}}
\begin{document}
\title{A characterization of horizontal visibility graphs and combinatorics on words}
\author{Gregory Gutin}
\affiliation{Department of Computer Science, Royal Holloway, University of London}
\author{Toufik Mansour}
\affiliation{Department of Mathematics, University of Haifa}
\author{Simone Severini}
\affiliation{Department of Physics and Astronomy, University College London}

\begin{abstract}
An Horizontal Visibility Graph (for short, HVG) is defined in association with
an ordered set of non-negative reals. HVGs realize a methodology in the
analysis of time series, their degree distribution being a good discriminator
between randomness and chaos [B. Luque, \emph{et al.}, \emph{Phys. Rev. E}
\textbf{80} (2009), 046103]. We prove that a graph is an HVG if and only if
outerplanar and has a Hamilton path. Therefore, an HVG is a noncrossing graph,
as defined in algebraic combinatorics [P. Flajolet and M. Noy, \emph{Discrete
Math.}, \textbf{204} (1999) 203-229]. Our characterization of HVGs implies a
linear time recognition algorithm.\textbf{ }Treating ordered sets as words, we
characterize subfamilies of HVGs highlighting various connections with
combinatorial statistics and introducing the notion of a visible pair. With
this technique we determine asymptotically the average number of edges of HVGs.

\end{abstract}
\maketitle

\section{Introduction}

A \emph{graph} is an ordered pair $G=\left(  V,E\right)  $, where $V$ is a set
of elements called \emph{vertices} and $E\subseteq V\times V-\{\{i,i\}$: $i\in
V\}$ is set of unordered pairs of vertices called \emph{edges}. Let $X=\left(
x_{i}\in\mathbb{R}^{\geq0}:i=1,2,...,n\right)  $ be an ordered set (or,
equivalently, a \emph{sequence}) of non-negative real numbers. The
\emph{horizontal visibility graph} (for short, \emph{HVG}) \cite{lu2} of $X$
is the graph $G=\left(  V,E\right)  $, with $V=X$ and an edge $x_{i}x_{j}$ if
$x_{i},x_{j}>x_{k}$ for every $i<k<j$. Clearly $x_{i}x_{j}\in E$ whenever
$j=i+1$. If the elements of $X$ are given with finite precision, we can always
take $X=\left(  x_{i}\in\mathbb{N}:i=1,2,...,n\right)  $. The HVG of $X$ is
also denoted by HVG$\left(  X\right)  $.

The term HVG also describes some diagrams employed in floorplanning and
channel routing of integrated circuits \cite{ho}; visibility graphs obtained
from polygons constitute an area of extensive study in computational geometry
(see, \emph{e.g.}, Ch. 15 in \cite{be}). Inspired by the geometric notion,
Lacasa \emph{et al.} \cite{lu1} introduced visibility graphs and shortly
afterwards their variations, HVGs \cite{lu2}. Further research on such graphs
is reported in \cite{lu3,la}. The principal idea of these papers is to
translate dynamical properties of time series into structural features of
graphs. It was shown in \cite{lu2} that the degree distribution of HVGs has a
role in discriminating between random and chaotic series obtained from
discrete dynamical systems. In \cite{la}, the authors use an algorithm for
HVGs to characterize and distinguish between correlated stochastic,
uncorrelated and chaotic processes.

The purpose of the present paper is to study combinatorial properties of HVGs.
We characterize the family of HVGs by proving that a graph is an HVG if and
only if is outerplanar and has a Hamilton path. This result allows us to use,
in the study of HVGs, results obtained for outerplanar graphs. Recognize
properties of a sequence $X$ via the properties of its HVG is a wide direction
of analysis. Treating ordered sets as words, we characterize subfamilies of
HVGs highlighting various connections with combinatorial statistics.

\section{Outerplanarity}

A \emph{drawing} of a graph $G$ is a function $f$ that maps each vertex $v\in
V(G)$ to a point $f(v)\in\mathbb{R}^{2}$ and each edge $uv\in E(G)$ to a curve
whose endpoints are $f(u)$ and $f(v)$. A graph is \emph{planar} if it has a
drawing so that any pair of edges can only intersect at their endpoints. A
\emph{face }is a bounded region of a planar graph. The \emph{infinite face} is
the outer one. A planar graph is \emph{outerplanar} if it has a drawing such
that each vertex is incident to the infinite face. A linear time recognition
algorithm for outerplanar graphs has been given by Mitchell \cite{mitch}.
Outerplanar graphs are used in the design of integrated circuits when the
terminals are on the periphery of the chip \cite{li}. In this section, we give
a characterization of HVGs with respect to outerplanarity. The
characterization is a consequence of the following two results:

\begin{lemma}
Every HVG is outerplanar and has a Hamilton path.
\end{lemma}

\begin{proof}
Let $G=$ HVG($X$), where $X=(x_{1},\ldots,x_{n})$. By definition, $G$ has a
Hamilton path $P=x_{1}x_{2}\ldots x_{n}$. Observe that there is no a pair
$x_{i}x_{j}$, $x_{s}x_{t}$ of edges such that $i<s<j<t$. Indeed, the existence
of such a pair implies that $x_{j}>x_{s}$ (due to $x_{i}x_{j}\in G$) and
$x_{j}<x_{s}$ (due to $x_{s}x_{t}\in G$), a contradiction. Thus, we can obtain
an outerplanar embedding of $G$ as follows: set as an interval of a horizontal
straight line with points being vertices of $P$ and depict each edge of $G$
outside $P$ as an arc above the interval for $P$ such that if $x_{i}%
x_{j},x_{s}x_{t}\in E(G)\setminus E(P)$ and $s<i<j<t$ then the arc
corresponding to $x_{i}x_{j}$ is situated below the arc corresponding to
$x_{s}x_{t}.$ Hence, we have proved that $G$ is outerplanar.
\end{proof}

\begin{lemma}
If a graph $H$ is outerplanar and has a Hamilton path, then $H$ is an HVG.
\end{lemma}

\begin{proof}
Let $H$ have a Hamilton path $v_{1}v_{2}\ldots v_{n}$. As $H$ is outerplanar,
all vertices of the Hamilton path belong to the infinite face in an
outerplanar embedding of $H$ into plane. Observe that $H$ has no crossing
edges, i.e., pairs of edges $v_{i}v_{j}$, $v_{s}v_{t}$ such that $i<s<j<t.$ We
say that an edge $v_{i}v_{j}$, $i<j$, of $H$ is of \emph{nesticity} 0 if
$j=i+1$. We say that an edge $v_{i}v_{j}$, $i<j$, of $H$ is of
\emph{nesticity} $k$ if $k$ is the minimum nonnegative integer such that if
$v_{s}v_{t}\in E(H)$ and $i<s<t<j$ then $v_{s}v_{t}$ is of nesticity $p$,
where $p<k.$ Nesticity of each edge is well-defined as there are no crossing
edges. The \emph{nesticity} of a vertex $v_{i}$ is the maximum nesticity of an
edge incident to $v_{i}.$

Now we will construct an ordered set $X=(x_{1},\ldots,x_{n})$ such that $H=$
HVG($X$) using the following algorithm which consists of two stages. In the
first stage, we determine the nesticity of each edge of $H$. To initialize,
set the nesticity of each edge $x_{i}x_{i+1}$ to 0. Now for each $q$ from 2 to
$n-1$ consider all edges of $H$ of the form $x_{i}x_{i+q}$ and set the
nesticity of $x_{i}x_{i+q}$ to $p+1$, where $p$ is the maximum nesticity of an
edge $x_{j}x_{k}$ with $i\leq j\leq k\leq i+q$ and $x_{j}x_{k}\neq
x_{i}x_{i+q}.$
In the second stage, compute the nesticity of each vertex $v_{i}$ of $H$ by
considering all edges incident to $v_{i}$ and set $x_{i}$ to the nesticity of
$v_{i}$.

It is easy to see, by induction on the nesticity of edges, that indeed $H=$
HVG($X$) (we use the fact that $H$ has no crossing edges).
\end{proof}

\begin{theorem}
\label{out} A graph is an HVG if and only if it is outerplanar and has a
Hamilton path.
\end{theorem}

Since recognizing an outerplanar graph and determining if it has a Hamilton
path are tasks that can be done in linear time (see Mitchell \cite{mitch} and
Lingas \cite{lin}, respectively), we have following result:

\begin{corollary}
We can determine in linear time (with respect to the number of vertices) if a
graph is an HVG.
\end{corollary}

We conclude this section with a further characterization. A \emph{noncrossing
graph} \cite{fn} with $n$ vertices is a graph drawn on $n$ points numbered in
counter-clockwise order on a circle such that the edges lie entirely within
the circle and do not cross each other. It is useful to point out that these
objects are HVGs because of Theorem \ref{out}:

\begin{corollary}
An HVG is a noncrossing graph.
\end{corollary}

\section{Unimodal HVGs}

We characterize here the HVGs with minimum number of edges. The \emph{degree}
of a vertex $i$ is $d(i):=\left\vert \{j:\{i,j\}\in E\}\right\vert $. The
\emph{degree sequence} is the unordered multiset of the degrees. HVGs are not
characterized by their degree sequence. In other words, there are
nonisomorphic HVGs with the same degree sequence. For example, given
$X=(x_{i}:1,2,...,5)$, let $x_{1}=x_{3}$, $x_{2}<x_{1}$ and $x_{2}=x_{4}%
=x_{5}$. The degree sequence of HVG$\left(  X\right)  $ is $\{3,2,3,2,2\}$.
The same degree sequence is associated to the graph HVG$\left(  X^{\prime
}\right)  $, where $x^{\prime}_{1}=x^{\prime}_{4}$, $x^{\prime}_{2}<x^{\prime
}_{1}$ and $x^{\prime}_{2}=x^{\prime}_{3}=x^{\prime}_{5}$. Let $\delta(G)$ and
$\Delta(G)$ be the \emph{minimum} \emph{degree }and \emph{maximum degree} of a
graph $G$, respectively. Theorem \ref{out} implies the following:

\begin{proposition}
If $G$ is an HVG then $\delta(G)=1$ or $2$ and $\Delta(G)\leq n-1$. If
$\Delta(G)= n-1$, then there is only one vertex of maximum degree.
\end{proposition}

A real-valued function $f$ is \emph{unimodal} if there exists a value $m$ such
that $f\left(  x\right)  $ is monotonically increasing for $x\leq m$ and
monotonically decreasing for $x\geq m$. Thus, $f\left(  m\right)  $ is the
maximum of $f$ and its only local maximum. If $X$ is unimodal then HVG($X$) is
said to be \emph{unimodal}.

\begin{proposition}
An HVG is unimodal if and only if it is a path.
\end{proposition}

We associate a word of length $n-1$, called a \emph{difference}, to a set with
$n$ elements. The letters of the word are from the alphabet $\{0,+,-\}$. Each
letter corresponds to a pair of adjacent elements. If $x_{i}=x_{i+1}$ the
letter corresponding to the pair $(x_{i},x_{i+1})$ is $0$. When $x_{i}%
<x_{i+1}$ (resp. $x_{i}<x_{i+1}$) then the letter for the pair is $+$ (resp.
$-$). For example, $X=(1,4,8,8,2,4)$ gives the difference $++0-+$. Let us
observe that the difference $D$ of a unimodal HVG graph does contain the
\emph{pattern} $\left(  -,+\right)  $, \emph{i.e.}, it does not contain a
subword $-P+$, where $P$ is an arbitrary subword of $D$.

\begin{proposition}
The number of HVGs isomorphic to the $n$-path and associated to different
ordered sets of cardinality $n$ (without taking into account the actual values
of the single elements) is exactly $a(n)=2^{n-1}(n+2)$.
\end{proposition}

\begin{proof}
We prove that the number of differences of length $n$ without the pattern
$\left(  -,+\right)  $ is exactly $a(n)$. Clearly, $a_{0}=1$ and $a_{1}=3$.
Since each such a word $x=x_{1}x_{2}\cdots x_{n}$ of length $n$ can be written
as $x=+x^{\prime}$, $x=0x^{\prime}$, or $x=-x^{\prime\prime}$ then
$a_{n}=2a_{n-1}+b_{n}$, where $b_{n}$ denotes the number words in $B_{n}$ of
length $n$ on the alphabet $\{0,+,-\}$ that do not contain the pattern
$\left(  -,+\right)  $ and its leftmost letter is $-$. Since any word in
$B_{n}$ has leftmost letter $-$ and there is no pattern $(-,+)$ in the word,
any letter which is not the leftmost has two possibilities: $0$ or $-$. This
implies that $b_{n}=2^{n-1}$ and, hence, $a_{n}$ satisfies the recurrence
relation $a_{n}=2a_{n-1}+2^{n-1}$ with initial conditions $a_{0}=1$ and
$a_{1}=3$. Solving this recurrence relation gives the formula for $a_{n}$.
\end{proof}

\bigskip

The difference does not uniquely specify the HVG. In fact, there are
nonisomorphic HVGs with the same difference. For example, the sets $\left(
5,4,3,5\right)  $ and $(5,4,3,4)$ have the same difference $--+$, but the HVGs
associated to these sets have five and four edges, respectively.

\section{Maximal HVGs}

The \emph{triangulation of a }(convex) \emph{polygon} is a planar graph
obtained by partitioning the polygon into disjoint triangles such that the
vertices of the triangles are chosen from the vertices of the polygon. An
outerplanar graph is \emph{maximal outerplanar} if it is not possible to add
an edge such that the resulting graph is still outerplanar. A maximal
outerplanar graph can be viewed as a triangulation of a polygon. By Theorem
\ref{out},

\begin{corollary}
\label{ccxx} The maximum number of edges in an HVG on $n$ vertices is $2n-3$
and the graph is maximal outerplanar.
\end{corollary}

In computational geometry, the \emph{visibility graph} \cite{ta}\emph{\ }of a
polygon with $n$ angles is obtained by constructing a graph on $n$ vertices,
each vertex of the graph representing an angle of the polygon, and each edge
of the graph joining only those pairs of vertices that represent visible pairs
of angles in the polygon. A polygon in the plane is called \emph{monotone}
with respect to a straight line $L$, if every of the lines orthogonal to $L$
intersects the polygon at most twice. By a result of ElGindy \cite{elg} and
Theorem \ref{out}, we have an observation relating HVGs to a special classes
of visibility graphs:

\begin{corollary}
An HVG on $n$ vertices and $2n-3$ edges is the visibility graph of a monotone polygon.
\end{corollary}

A characterization of HVGs with maximal number of edges is given with
Corollary \ref{co2}. In fact, the characterization is easy once established a
connection between HVGs and combinatorics on words.

\section{Words}

We denote by $[k]^{n}$ the set of all words of length $n$ over an alphabet
$[k]=\{1,2,\ldots,k\}$. Each word $x=x_{1}x_{2}\cdots x_{n}$ defines an
ordered set $X$ by $X=\{x_{1},x_{2},\ldots,x_{n}\}$ and conversely. In order
to describe the edges of an HVG with respect to words, we need the following definition:

\begin{definition}
\label{defw1} Let $x=x_{1}x_{2}\cdots x_{n}$ be any word in $[k]^{n}$. We say
that the pair $(x_{i},x_{j})$ with $i+1\le j$ is \emph{visible} if
$x_{i+1},x_{i+2},\ldots,x_{j-1}<\min\{x_{i},x_{j}\}$. Clearly, $(x_{i}%
,x_{i+1})$ is a visible pair. We denote the number of visible pairs in $x$ by
$vis(x)$.
\end{definition}

For instance, if $x=21232143112112$ is a word in $[4]^{14}$ then $(x_{1}%
,x_{3})$, $(x_{4},x_{7})$, $(x_{5},x_{7})$, $(x_{8},x_{11})$ and
$(x_{11},x_{14})$ are the visible pairs of $x$, in addition to the $13$ edges
of $P_{14}$. Thus $vis(x)=18$.

By the definition of HVG and visibility of pairs in words we can state the
following result.

\begin{theorem}
Let $X=\{x_{1},x_{2},\ldots,x_{n}\}$ be an ordered set of $n$ elements and let
$k=\max X$. The HVG of $X$, $G=$HVG$\left(  X\right)  $, can be represented
uniquely as a word $x=x_{1}x_{2}\cdots x_{n}$ with a set $E$ of visible pairs,
where each edge $\{i,j\}$ in the graph $G$ corresponds to visible pair
$(x_{i},x_{j})$.
\end{theorem}

Let $A=(x_{i},x_{j})$ and $B=(x_{i^{\prime}},x_{j^{\prime}})$ be two visible
pairs in a word $x=x_{1}x_{2}\cdots x_{n}\in\lbrack k]^{n}$ with $i\leq
i^{\prime}$. Clearly, $i<j$ and $i^{\prime}<j^{\prime}$. From the above
definition, we obtain that either $i\leq i^{\prime}\leq j^{\prime}\leq j$ or
$i<j\leq i^{\prime}<j^{\prime}$. In such cases, we say that the pair $A$
\emph{covers} the pair $B$ and that $A$ and $B$ are \emph{disjoint} pairs, respectively.

\begin{fact}
If $x$ is the word realizing an HVG then in any two disjoint pairs $A,B\in x$
either $A$ covers $B$ or $B$ covers $A$.
\end{fact}

How many words are there in $[k]^{n}$ with a fixed number of visible pairs?

\begin{definition}
We denote the generating function for the number of words $x\in\lbrack k]^{n}$
according to the number of visible pairs in $x$ by $F_{k}(x,q)$, that is,
\[
F_{k}(x,q)=\sum_{n\geq0}x^{n}\sum_{x\in\lbrack k]^{n}}q^{vis(x)}.
\]
Similarly, we denote the generating function for the number of words
$x\in\lbrack k-1]^{n}$ according to the number of visible pairs in $xk$
(respectively, $kx$) by $L_{k}(x,q)$ (respectively, $R_{k}(x,q)$). Also, we
denote the generating function for the number of words $x\in\lbrack k-1]^{n}$
according to the number of visible pairs in $kxk$, which are not equal to
$(k,k)$, by $M_{k}(x,q)$.
\end{definition}

Note that each word $x$ in $[k]^{n}$ can be decomposed either as

\begin{itemize}
\item $x\in\lbrack k-1]^{n}$, that is, $x$ does not contain the letter $k$,

\item $x=x^{(1)}kx^{(2)}\cdots kx^{(m+1)}$, where $x^{(j)}$ is a word in
$[k-1]^{i_{j}}$ with $i_{1}+\cdots+i_{m+1}+m=n$.
\end{itemize}

Thus, rewriting these rules in terms of generating functions we obtain the
following lemma.

\begin{lemma}
For all $k\geq1$,
\[
F_{k}(x,q)=F_{k-1}(x,q)+\frac{xR_{k}(x,q)L_{k}(x,q)}{1-xqM_{k}(x,q)},
\]
and for all $k\geq2$,
\begin{align*}
M_{k}(x,q)  &  =M_{k-1}(x,q)+\frac{xq^{2}M_{k-1}(x,q)}{1-xqM_{k-1}(x,q)}\\
R_{k}(x,q)  &  =\frac{R_{k-1}(x,q)}{1-xqM_{k-1}(x,q)}\\
L_{k}(x,q)  &  =R_{k}(x,q).
\end{align*}
Also $F_{0}(x,q)=M_{1}(x,q)=R_{1}(x,q)=L_{1}(x,q)=1$.
\end{lemma}

The above lemma together with induction give the following result.

\begin{theorem}
\label{thc1} Let $k\geq1$. The generating function $F_{k}(x,q)$ is given by
\[
F_{k}(x,q)=1+\sum_{j=1}^{k}\frac{1-xqM_{j}(x,q)}{\prod_{i=1}^{j}%
(1-xqM_{i}(x,q))^{2}},
\]
where $M_{k}(x,q)$ satisfies the recurrence relation
\[
M_{k}(x,q)=M_{k-1}(x,q)+\frac{xq^{2}M_{k-1}(x,q)}{1-xqM_{k-1}(x,q)},
\]
with initial condition $M_{1}(x,q)=1$.
\end{theorem}

For instance, the theorem gives%
\[
F_{1}(x,q)=1+\frac{x}{1-xq}%
\]
and
\[
F_{2}(x,q)=1+\frac{x((1-xq)^{2}+1-x^{2}q^{3})}{(1-xq)((1-xq)^{2}-x^{2}q^{3})}.
\]
Note that it appears to be hard to derive an explicit formula for the
generating functions $M_{k}(x,q)$ and $F_{k}(x,q)$. But we can use the result
for studying the total number of visible pairs in all words in $[k]^{n}$.

\begin{theorem}
\label{thc2 copy(1)}The generating function for the total number of visible
pairs in all words in $[k]^{n}$ ($k\geq1$) is
\begin{align*}
F_{k}^{\prime}(x)  &  =\sum_{n\geq0}x^{n}\sum_{x\in\lbrack k]^{n}}vis(x)\\
&  =2x^{2}\sum_{j=1}^{k}\frac{\sum_{i=1}^{j}\frac{1-(i-1)x+x\sum_{\ell
=1}^{i-1}\frac{2-(2\ell-1)x}{1-(\ell-1)x}}{(1-(i-1)x)(1-ix)}}%
{(1-(j-1)x)(1-jx)}\\
&  -x^{2}\sum_{j=1}^{k}\frac{1-(j-1)x+x\sum_{i=1}^{j-1}\frac{2-(2i-1)x}%
{1-(i-1)x}}{(1-(j-1)x)^{2}(1-jx)^{2}}.
\end{align*}

\end{theorem}

\begin{proof}
We define%
\[
F_{k}^{\prime}(x)=\frac{d}{dq}F_{k}(x,q)\mid_{q=1}%
\]
and%
\[
M_{k}^{\prime}(x)=\frac{d}{dq}M_{k}(x,q)\mid_{q=1}.
\]
Theorem \ref{thc1} with induction on $k$ implies that $F_{k}(x,1)=\frac
{1}{1-kx}$ and $M_{k}(x,1)=\frac{1}{1-(k-1)x}$. By differentiating respect to
$q$ and by using the expressions $F_{k}(x,1)$ and $M_{k}(x,1)$, Theorem
\ref{thc1} gives
\begin{align*}
F_{k}^{\prime}(x) &  =-x^{2}\sum_{j=1}^{k}\frac{M_{j}+M_{j}^{\prime}}%
{\prod_{i=1}^{j}(1-xM_{i})^{2}}\\
&  +2x^{2}\sum_{j=1}^{k}\frac{1-xM_{j}}{\prod_{i=1}^{j}(1-xM_{i})^{2}}%
\sum_{i=1}^{j}\frac{M_{i}+M_{i}^{\prime}}{1-xM_{i}},
\end{align*}
where $M_{j}=\frac{1}{1-(j-1)x}$ and $M_{j}^{\prime}$ is given by the
recurrence relation
\begin{align*}
M_{k}^{\prime} &  =\frac{(1-(k-2)x)^{2}}{(1-(k-1)x)^{2}}M_{k-1}^{\prime}\\
&  +\frac{x(2-(3k-3)x)}{(1-(k-2)x)(1-(k-1)x)},
\end{align*}
with the initial condition $M_{1}^{\prime}=0$. Hence, by induction on $k$ we
obtain that%
\[
M_{k}^{\prime}=\frac{x}{(1-(k-1)x)^{2}}\sum_{j=1}^{k-1}\frac{2-(2j-1)x}%
{1-(j-1)x}.
\]
Plugging this into the equation of $F_{k}^{\prime}$ we have the result. For
instance, 
\end{proof}

\begin{itemize}
\item $F_{1}^{\prime}(x)=\frac{x^{2}}{(1-x)^{2}}$, 

\item $F_{2}^{\prime}(x)=\frac{(4-3x)x^{2}}{(1-x)(1-2x)^{2}}$, 

\item $F_{3}^{\prime}(x)=\frac{(10x^{2}-22x+9)x^{2}}{(1-x)(1-2x)(1-3x)^{2}}$. 
\end{itemize}

Theorem \ref{thc2} shows that the function $F_{k}^{\prime}(x)$ has a pole of
degree two at $\frac{1}{k}$. Direct calculations show that
\begin{align*}
C &  =\lim_{x\rightarrow\frac{1}{k}}(1-kx)^{2}F_{k}^{\prime}(x)\\
&  =\lim_{x\rightarrow\frac{1}{k}}\frac{x^{2}\left(  1-(k-1)x+x\sum
_{i=1}^{k-1}\frac{2-(2i-1)x}{1-(i-1)x}\right)  }{(1-(k-1)x)^{2}}\\
&  =\frac{1}{k}\sum_{i=1}^{k}\frac{2k+1-2i}{k+1-i}\\
&  =2-\frac{\Psi(k+1)+\gamma}{k},
\end{align*}
where $\Psi(x)$ is the digamma function and $\gamma$ is Euler's constant.
Thus, asymptotically, the total number of visible pairs in all words in
$[k]^{n}$ ($k\geq1$) is given by $Cnk^{n}$. This means that the average number
of visible pairs in all words in $[k]^{n}$ is given by $Cn$.

\begin{corollary}
\label{thc2}The average number of edges in an HVG is%
\[
\left(  2-(\Psi(k+1)+\gamma)/k\right)  n,
\]
when $X$ is an ordered subset of $n$ elements from $[k]$ and $n\rightarrow
\infty$.
\end{corollary}

The corollary shows that the number of maximal edges is in fact $2n-3$, where
$X$ contains $n$ elements (cfr. Section III). This can be obtained easily from
our representation as words:

\begin{corollary}
\label{co2}Let $G$ be an HVG with $2n-3$ edges. Then%
\[
X=\{\ldots,8,6,4,2,1,3,5,6,7,9,\ldots\}
\]
or%
\[
X=\{\ldots,9,7,5,3,1,2,4,6,8,\ldots\}.
\]

\end{corollary}

\begin{proof}
Without loss of generality, we can assume that $X$ has $n$ different numbers.
Assume $\pi=\pi_{1}\pi_{2}\cdots\pi_{n}$ is a permutation on $[n]$ with a
maximal number of visible pairs. The pair $(\pi_{1},\pi_{n})$ is visible,
hence either $\pi_{1}+1=\pi_{n}=n$ or $\pi_{n}+1=\pi_{1}=n$. Delete $n$ from
$\pi$ and denote the resulting permutation by $\pi^{\prime}$ By induction on
$n$, we obtain that $\pi$ has the form $\pi=\cdots864213579\cdots$ or
$\pi=\cdots975312468\cdots$, as claimed.
\end{proof}

\section{Permutations}

We denote by $S_{n}$ the set of all permutations on $[n]$. Each permutation
$\pi=\pi_{1}\pi_{2}\cdots\pi_{n}$ defines an ordered set $X$ by $X=\{\pi
_{1},\pi_{2},\ldots,\pi_{n}\}$ and conversely. By definition \ref{defw1}, we
can describe the edges of an HVG with respect to permutations (a permutation
is in fact a word without repetitions). Define $S_{n}(q)$ to be the generating
function for the number of permutations $\pi$ on $[n]$ according to the number
of visible pairs in $\pi$ by $S_{n}(q)$, that is,
\[
S_{n}(q)=\sum_{\pi\in S_{n}}q^{vis(\pi)}.
\]
For example, in $S_{2}$ there are two permutations $12$ and $21$. Thus
$S_{2}(q)=2q$. In $S_{3}$ there are $6$ permutations $123$, $132$, $213$,
$231$, $312$ and $321$. Thus $S_{3}(q)=4q^{2}+2q^{3}$.

Now, let us find an explicit formula for $S_{n}(q)$. Define $S_{n,j}$ to be
the set of permutations $\pi=\pi_{1}\pi_{2}\cdots\pi_{n}$ in $S_{n}$ such that
$\pi_{j}=1$. Let $\pi\in S_{n}$ and define $\pi^{\prime}$ to be the
permutation obtained from $\pi$ by deleting the letter $1$ and by decreasing
each letter by $1$. Let us write an equation for $S_{n}(q)$. From the
definitions, we have that, if $j=1,n$ then the set permutations $S_{n,j}$ is
counted by $qS_{n-1}(q)$, and if $2\leq j\leq n-1$ then the set of
permutations $S_{n,j}$ is counted by $q^{2}S_{n-1}(q)$. Thus, for all $n\geq
2$,
\[
S_{n}(q)=2qS_{n-1}(q)+(n-2)q^{2}S_{n-1}(q),
\]
which implies that%
\[
S_{n}(q)=(2q)^{n-1}\prod_{j=2}^{n}\left(  1+\frac{j-2}{2}q\right)  .
\]
Using the fact that the unsigned Stirling numbers $s(n,j)$ of the first kind
satisfy the relation
\[
(1+x)\cdots(1+(n-1)x)=\sum_{j=0}^{n}s(n,n-j)x^{j},
\]
we obtain that
\begin{align*}
S_{n}(q) &  =(2q)^{n-1}\prod_{j=1}^{n-2}\left(  1+j\frac{q}{2}\right)  \\
&  =(2q)^{n-1}\sum_{j=0}^{n-1}s(n-1,n-1-j)\frac{q^{j}}{2^{j}}\\
&  =\sum_{j=0}^{n-1}2^{n-1-j}s(n-1,n-1-j)q^{n-1+j}.
\end{align*}
Hence, we can state the following result:

\begin{theorem}
\label{pp1} Let $n\geq2$. The number of permutations $\pi$ with exactly
$n-1+j$, $0\leq j\leq n-1$, visible pairs is given by
\[
2^{n-1-j}s(n-1,n-1-j),
\]
where $s(n-1,n-1-j)$ is the unsigned Stirling number of the first kind.
\end{theorem}

\begin{corollary}
The average number of edges in an HVG is%
\[
2n-\sum_{j=1}^{n}\frac{1}{j},
\]
when $X$ is an ordered subset of $n$ different elements.
\end{corollary}

\begin{proof}
By Theorem \ref{pp1} we have that the average number of edges in an HVG is%
\[
p_{n}=\frac{1}{n!}\sum_{j=0}^{n-1}(n-1+j)2^{n-1-j}s(n-1,n-1-j),
\]
when $X$ is an ordered subset of $n$ different elements. The expression
$p_{n}$ can be written as
\[
p_{n}=\frac{1}{n!}\sum_{j=0}^{n-1}(2n-2-j)2^{j}s(n-1,j).
\]
By the fact that
\[
x(x+1)\cdots(x+n-1)=\sum_{j=0}^{n}s(n,j)x^{j},
\]
we obtain
\[
\sum_{j=0}^{n}js(n,j)x^{j}=x(x+1)\cdots(x+n-1)\sum_{j=0}^{n-1}\frac{1}{x+j},
\]
which implies that%
\[
\sum_{j=0}^{n-1}s(n-1,j)2^{j}=n!
\]
and%
\[
\sum_{j=0}^{n-1}j2^{j}s(n-1,j)=n!\sum_{j=2}^{n}\frac{1}{j}.
\]
Hence,%
\[
p_{n}=2(n-1)-\sum_{j=2}^{n}\frac{1}{j}=2n-\sum_{j=1}^{n}\frac{1}{j},
\]
which completes the proof.
\end{proof}

\section{Conclusions}

We have characterized the family HVGs in terms of their combinatorial
properties. The characterization is useful because it gives an efficient
recognition algorithm. Moreover, it places the study of HVGs in a specific
mathematical context, related to a well-known class of graphs. As it is
originally observed in the literature, the potential importance of HVGs stems
in their use for describing properties of dynamical objects like time series.
Therefore, the main goal would be to determine the dynamical and structural
properties of an ordered set that are readable through the analysis of its
HVG. In this perspective, we have shown that combinatorics on words are a
useful tool. The connection suggests a number of natural open problems. For
example, it may be valuable to study HVGs and statistics on words and
forbidden subsequences. The mathematical scope turns out to be wider than the
original dynamical systems framework.

It would be interesting to characterize visibility graphs \cite{lu1,lu3},
which can be defined as follows. For an ordered set $X=\left(  x_{i}%
\in\mathbb{R}^{\geq0}:i=1,2,...,n\right)  $, its \emph{visibility graph}
VG($X$) has the vertex set $X$ and $x_{i}$, $x_{j}$ $(i<j)$ are adjacent if
for each $k$ with $i<k<j$ we have
\[
x_{k}<x_{j}+(x_{i}-x_{j})\frac{j-k}{j-i}.
\]
It is easy to see that for each $X$, HVG($X$) is a subgraph of HV($X$). Not
each visibility graph is outerplanar as HV($Y$), where $Y=(4,2,1,4)$ is
$K_{4}$, which is not outerplanar.

\bigskip

\emph{Acknowledgments. }We would like to thank Lucas Lacasa for useful
conversation about HVGs.

\end{document}